\documentclass[runningheads,a4paper]{llncs}

\usepackage{amssymb}
\usepackage{graphicx}

\usepackage{comment,todonotes}

\newcommand{\knote}[1]{\todo[inline, color=blue!20]{#1}}
\newcommand{\fnote}[1]{\todo[inline, color=green!20]{#1}}
\newcommand{\newknote}[1]{\todo[inline, color=red!20]{#1}}

\renewcommand{\knote}[1]{}
\renewcommand{\fnote}[1]{}
\renewcommand{\newknote}[1]{}

\usepackage{url}
\newcommand{\keywords}[1]{\par\addvspace\baselineskip
\noindent\keywordname\enspace\ignorespaces#1}

\usepackage{soul}

\usepackage{tikz}
\usetikzlibrary{matrix}

\usepackage{amssymb}
\usepackage{amsmath} 
\usepackage{semantic}
\usepackage{bussproofs}
\usepackage[lighttt]{lmodern}
\usepackage{listings}
\usepackage{xspace}

\def\etal{\emph{et al.}\@\xspace}
\def\cf{\emph{cf.}\@\xspace}
\def\ie{\emph{i.e.}\@\xspace}

\DeclareMathOperator{\HNF}{HNF}

\usepackage{thmtools}
\usepackage{thm-restate}

\newcommand{\US}{\normalfont\mathbf{U}_\Sigma}

\newcommand{\BS}{\normalfont\mathbf{B}_\Sigma}

\newcommand{\TP}{\normalfont\mathcal{T}_P}

\newcommand{\MP}{\normalfont\mathcal{M}_P}
\newcommand{\McP}{\normalfont\mathcal{M}'_P}

\usepackage[normalem]{ulem}

\usepackage[inline]{enumitem}

\usepackage[]{proofs}

\begin{document}

\mainmatter  

\title{Coinductive Soundness of Corecursive Type Class Resolution}

\titlerunning{Coinductive Soundness of Corecursive Type Class Resolution}

\author{Franti\v{s}ek Farka\inst{1,2}%
\and Ekaterina Komendantskaya\inst{3}\and Kevin Hammond\inst{2}}
\authorrunning{F. Farka, E. Komendantskaya, and K. Hammond}

\institute{University of Dundee, Dundee, Scotland\\
\and
University of St Andrews, St Andrews, Scotland\\
\email{\{ff32,kh8\}@st-andrews.ac.uk}\\
\and
Heriot-Watt University, Edinburgh, Scotland\\
\email{ek19@hw.ac.uk}\\
}

\toctitle{Coinductive Soundness of Corecursive Type Class Resolution}
\tocauthor{Franti\v{s}ek Farka, Ekaterina Komendantskaya, and Kevin Hammond}
\maketitle

\begin{abstract}
  Horn clauses and first-order resolution are commonly used to
  implement \emph{type classes} in Haskell.  Several
  corecursive extensions to type class resolution have recently been proposed,
  with the goal of allowing (co)recursive
  dictionary construction where resolution does not
  terminate.  This paper shows, for the first time, that corecursive type
  class resolution and its extensions are \emph{coinductively sound} with
  respect to the greatest Herbrand models of logic programs and that they
  are \emph{inductively unsound} with respect to the least Herbrand models.
  We establish incompleteness results for various fragments of the proof system.

\end{abstract}

\section{Introduction}
\label{sec:intro}

Type classes can be used to implement ad-hoc
polymorphism and overloading in functional languages. The approach originated in
Haskell~\cite{WadlerB89,toplas} and has been further developed in
dependently typed languages~\cite{GonthierZND11,DevrieseP11}. 
For example, it is convenient to define equality for all data
structures in a uniform way.  In Haskell, this is achieved by introducing the
equality class $\mathtt{Eq}$:

\begin{lstlisting}
class Eq x where
  eq :: Eq x => x -> x -> Bool
\end{lstlisting}

\noindent and then declaring any necessary instances of the class, e.g. for pairs and integers:

\begin{lstlisting}
instance (Eq x, Eq y) => Eq (x, y) where
  eq (x1, y1) (x2, y2) = eq x1 x2 && eq y1 y2
instance Eq Int where
  eq x y = primtiveIntEq x y
\end{lstlisting}

\noindent \emph{Type class resolution} is performed by the Haskell compiler and involves
checking whether all the instance declarations are valid. For example, the
following function triggers a check that $\mathtt{Eq~(Int,~
Int)}$ is a valid instance of type class $\mathtt{Eq}$:
\begin{lstlisting}
test :: Eq (Int, Int) => Bool
test = eq (1,2) (1,2)
\end{lstlisting}

\noindent It is folklore that type class instance resolution
resembles SLD-resolution from logic programming. 
The type class instance declarations above
could, for example, be viewed as the following two Horn clauses:
\begin{example}[Logic program $P_{Pair}$]\label{ex:pair}
\begin{alignat*}{3}
    \kappa_1 :~ &&  \mathtt{eq}(x),~\mathtt{eq}(y)
    & ~\Rightarrow \mathtt{eq}(\mathtt{pair}(x, y))\\
    \kappa_2 :~ &&                & ~\Rightarrow \mathtt{eq}(\mathtt{int})
\end{alignat*}
\end{example}
Then, given the query $\mathtt{?}~\mathtt{eq(pair(int,int))}$, SLD-resolution terminates
successfully with the following sequence of inference steps:
\abovedisplayskip=0.2em
\belowdisplayskip=0.2em
\begin{align*}
    \small
    \mathtt{eq}(\mathtt{pair}(\mathtt{int}, \mathtt{int})) \to_{\kappa_1}
    \mathtt{eq}(\mathtt{int}), \mathtt{eq}(\mathtt{int}) \to_{\kappa_2}
    \mathtt{eq}(\mathtt{int}) \to_{\kappa_2}
    \emptyset
\end{align*}
\noindent The proof witness $\kappa_1 \kappa_2 \kappa_2$  (called a ``dictionary'' in Haskell) is constructed by the Haskell
compiler. This is treated internally
as an executable function.

\vspace{1em}

Despite the apparent similarity of type class syntax and type class resolution
to Horn clauses and SLD-resolution they are not, however, identical. At a
syntactic level,  type class instance declarations correspond to a restricted form
of Horn clauses, namely ones that:
\begin{enumerate*}[label=(\roman*)]
    \item \label{r:nonoverlap}  do not \emph{overlap} (\ie
        whose heads do not unify); and that
    \item \label{r:noexist} do not contain existential variables (\ie
        variables that occur in the bodies but not in the heads of the clauses). At
        an algorithmic level,
    \item \label{r:match} type class resolution corresponds to SLD-resolution
        in which unification is restricted to term-matching.
\end{enumerate*}
Assuming there is a clause $B_1, \ldots B_n \Rightarrow A'$, 
then a
query $?~A'$ can be resolved with this clause only if $A$ can be matched against $A'$, \ie if a
substitution $\sigma$ exists such that  $A = \sigma A'$.  In comparison,
SLD-resolution incorporates \emph{unifiers}, as well as \emph{matchers}, \ie it 
also proceeds to resolve the above query and clause in all the cases where $\sigma A =
\sigma A'$ holds.

These restrictions guarantee that type class inference computes the
\emph{principal} (most general) type.  Restrictions
\ref{r:nonoverlap} and \ref{r:noexist}  amount to deterministic inference by
resolution, in which only one derivation is possible for every query.
Restriction \ref{r:match} means that no substitution is applied to a query during 
inference, \ie we prove the query in an implicitly universally quantified form.
It is common knowledge that (as with SLD-resolution) type class
resolution is \emph{inductively sound}, \ie that it is sound relative to the least
Herbrand models of logic programs~\cite{Lloyd87}.  Moreover, in Section~\ref{sec:ind_sound} we establish, for the first time, that  it is also
\emph{universally inductively sound}, \ie that if a formula $A$ is proved by type
class resolution, every ground instance of $A$ is in the least Herbrand model of
the given program. 
%
In contrast to SLD-resolution, however, type class resolution is \emph{inductively incomplete},
\ie it is incomplete relative to least Herbrand models, even for the class of
Horn clauses that is restricted by conditions 
\ref{r:nonoverlap} and 
\ref{r:noexist}.  
For example, given a clause
$\Rightarrow \mathtt{q}(\mathtt{f}(x))$ and  a query $\mathtt{?}~\mathtt{q}(x)$,
SLD-resolution is able to find a proof (by instantiating $x$ with
$\mathtt{f}(x)$), but type class resolution fails.
\newknote{I think the above example is really needed here, otherwise the meaning of the incompleteness claim is unclear }
\fnote{The reviewer 2 asked for a type-class example, I am not aware of any}
L\"ammel and Peyton Jones have suggested~\cite{LammelJ05} an extension to type class
resolution that accounts for some non-terminating cases of type class
resolution.  Consider, for example, the following  mutually defined data structures:

\vskip-0.5em
\begin{lstlisting}
data OddList a   =  OCons a (EvenList a)
data EvenList a  =  Nil | ECons a (OddList a)
\end{lstlisting}

\noindent which give rise to the following instance declarations for the \texttt{Eq} class:
 
\begin{lstlisting} 
instance (Eq a, Eq (EvenList a)) =>  Eq (OddList a) where
    eq (OCons x xs) (OCons y ys) =  eq x y && eq xs ys

instance (Eq a, Eq (OddList a)) => Eq (EvenList a) where
    eq Nil          Nil          =  True
    eq (ECons x xs) (ECons y ys) =  eq x y && eq xs ys
    eq _            _            =  False
\end{lstlisting}

\noindent 
The \texttt{test} function below triggers type class resolution in the Haskell compiler:
\begin{lstlisting}
test :: Eq (EvenList Int) => Bool
test = eq Nil Nil
\end{lstlisting}
%
\noindent However, inference by resolution does not terminate in this case. Consider the
Horn clause representation of the type class instance declarations:
\begin{example}[Logic program $P_{EvenOdd}$]\label{ex:oddeven}
\begin{alignat*}{3}
    \kappa_1 :~ & \mathtt{eq}(x),~\mathtt{eq}(\mathtt{evenList}(x))
        && ~\Rightarrow \mathtt{eq}(\mathtt{oddList}(x))\\
    \kappa_2 :~ & \mathtt{eq}(x),~\mathtt{eq}(\mathtt{oddList}(x))
        && ~\Rightarrow \mathtt{eq}(\mathtt{evenList}(x))\\
    \kappa_3 :~ & && ~\Rightarrow \mathtt{eq}(\mathtt{int})
\end{alignat*}
\end{example}
The non-terminating resolution trace is given by:
\abovedisplayskip=0.2em
\belowdisplayskip=0.2em
\begin{align*}
    \underline{\mathtt{eq}(\mathtt{evenList}(\mathtt{int}))} \to_{\kappa_2} 
    \mathtt{eq}(\mathtt{int}),~\mathtt{eq}(\mathtt{oddList}(\mathtt{int}))
        \to_{\kappa_3}
    \mathtt{eq}(\mathtt{oddList}(\mathtt{int})) \\
    \to_{\kappa_1}
    \mathtt{eq}(\mathtt{int}),~\mathtt{eq}(\mathtt{evenList}(\mathtt{int}))
        \to_{\kappa_3}
    \underline{\mathtt{eq}(\mathtt{evenList}(\mathtt{int}))} \to_{\kappa_2}
    \dots
\end{align*}
A goal $\mathtt{eq}(\mathtt{evenList}(\mathtt{int}))$ is simplified using the
clause $\kappa_2$ to goals $\mathtt{eq}(\mathtt{int})$ and
$\mathtt{eq}(\mathtt{oddList}(\mathtt{int}))$. The first of these is discarded using
the clause $\kappa_3$. Resolution continues using clauses $\kappa_1$ and
$\kappa_3$, resulting in the original goal
$\mathtt{eq}(\mathtt{evenList}(\mathtt{int}))$. It is easy to see that such
a process could continue infinitely and that this goal constitutes a \emph{cycle} (underlined above). 

As suggested by L\"ammel and Peyton Jones~\cite{LammelJ05}, the compiler can obviously terminate the
infinite inference process as soon as it detects the  underlined cycle.  Moreover, it
can also construct the corresponding proof witness in a form of a recursive function.
For the  example above, such a function is given by the fixed point
term $ \nu \alpha. \kappa_2 \kappa_3 (\kappa_1 \kappa_3 \alpha)$, where
$\nu$ is a fixed point operator.
The intuitive reading of such a proof is that an infinite proof of the query
$\mathtt{?}~\mathtt{eq\ (evenList (int))}$ exists, and that its shape is fully
specified by the recursive proof witness function above.
We 
say that the proof
is given by \emph{corecursive type class resolution}.

Corecursive type class resolution is not inductively sound. For example,  the formula 
$\mathtt{eq}(\mathtt{evenList}(\mathtt{int}))$ is not in the least
Herbrand model of the corresponding logic program.  However, as we prove in Section~\ref{sec:coind_sound}, it is
\emph{(universally) coinductively sound}, \ie it is sound relative to the greatest
Herbrand models. For example, $\mathtt{eq}(\mathtt{evenList}(\mathtt{int}))$
is in the greatest Herbrand model of the program $P_{EvenOdd}$.
%
Similarly to the inductive case, corecursive type class resolution is
coinductively incomplete. Consider the clause 
$\kappa_{inf} : \mathtt{p}(x) \Rightarrow \mathtt{p}(\mathtt{f}(x))$.
This clause may be given an interpretation by the greatest (complete)
Herbrand models. However, corecursive type class resolution does not yield
infinite proofs.
\newknote{The above would be hard for understanding if you remove the example of inductive incompleteness}

Unfortunately, this simple method of cycle detection does not work for all non-terminating programs.  Consider
the following example, which defines a data type $\mathtt{Bush}$
(for bush trees), and its corresponding instance for \texttt{Eq}:

\begin{lstlisting}
data Bush a = Nil | Cons a (Bush (Bush a))
instance Eq a, Eq (Bush (Bush a)) => Eq (Bush a) where { ... }
\end{lstlisting}
\noindent Here, type class resolution does not
terminate.  However, it does not exhibit cycles. Consider the Horn clause
translation of the problem:
\begin{example}[Logic program $P_{Bush}$]\label{ex:bush}
    ~
\vskip-1.5em
\begin{alignat*}{3}
    \kappa_1 :~ & && \Rightarrow \mathtt{eq}(\mathtt{int})\\
    \kappa_2 :~ & \mathtt{eq}(x),~\mathtt{eq}(\mathtt{bush}(\mathtt{bush}(x)))
        && \Rightarrow \mathtt{eq}(\mathtt{bush}(x))
\end{alignat*}
\end{example}
\vskip-0.5em
\noindent The derivation below  shows that no cycles arise when we resolve the
query $\mathtt{?}~\mathtt{eq}(\mathtt{bush}(\mathtt{int}))$ against the 
program $P_{Bush}$:
\abovedisplayskip=0.2em
\belowdisplayskip=0.2em
\begin{align*}
\mathtt{eq(bush(int))}
\to_{\kappa_2} \mathtt{eq(int), eq(bush(bush(int))}
\to_{\kappa_1} \ldots 
\to_{\kappa_2} \\ \mathtt{eq(bush(int)), eq(bush(bush(bush(int)))} 
\to_{\kappa_1} 
\dots
\end{align*}

\newknote{be careful with terminology: now mu became nu, should we have ``corecursive'' proof witness?}
\fnote{We only consistently replaced $\mu$ by $\nu$ as the Reviewer 1 suggested.
I do not see a reason why not ``corecursive''}

\noindent
Fu \etal \cite{FuKSP16} have recently introduced an extension to corecursive
type class resolution that allows implicative queries  to be proved by corecursion
and uses the recursive proof witness construction.
For example, in  the above program the Horn formula 
$\mathtt{eq}(x) \Rightarrow \mathtt{eq}(\mathtt{bush}(x))$
can be (coinductively) proven with the recursive proof witness 
$\kappa_3 = \nu \alpha . \lambda \beta. \kappa_2 \beta (\alpha (\alpha \beta))$. 
If we add this Horn clause as a third clause to our program, we obtain a proof
of $\mathtt{eq(bush(int))}$ by applying $\kappa_3$ to $\kappa_1$.  In this case, it
is even more challenging to understand whether the proof $\kappa_3 \kappa_1$ of
$\mathtt{eq(bush(int))}$ is indeed sound: whether inductively, coinductively or in any
other sense.  In Section~\ref{sec:ext_coind_sound}, we establish, for the first
time, \emph{coinductive soundness} for proofs of such implicative queries, 
relative to the greatest Herbrand models of logic programs.
Namely, we determine that proofs that are obtained by extending the proof context with  coinductively proven Horn clauses
(such as $\kappa_3$ above) are coinductively sound but inductively unsound.
This result completes our study of the semantic properties of corecursive type class resolution.
Sections~\ref{sec:ind_sound} and~\ref{sec:ext_coind_sound} summarise our arguments concerning the inductive and coinductive incompleteness of corecursive type class resolution. 

\subsubsection{Contributions of this paper}

By presenting the described results, we answer 
three research questions:
\begin{enumerate}[label=(\arabic*)]
    \item \label{rq:1} whether type class resolution and its two recent
        corecursive extensions~\cite{FuKSP16,LammelJ05} are sound relative to the
        standard (Herbrand model) semantics of logic programming;
    \item \label{rq:2} whether these new extensions are indeed ``corecursive'',
        i.e. whether they are better modelled by the greatest Herbrand model semantics rather
        than by the least Herbrand model semantics; and
    \item \label{rq:3} whether the context update technique given
        in~\cite{FuKSP16} can be reapplied to logic programming and can be
        re-used in its corecursive dialects such as CoLP~\cite{SimonBMG07} and
        CoALP~\cite{KomendantskayaJ15} or, even broader, can be incorporated
        into program transformation techniques~\cite{AngelisFPP15}.
\end{enumerate}

\noindent
We answer questions \ref{rq:1} and \ref{rq:2} in the affirmative.  The answer to
question \ref{rq:3} is less straightforward.  The way the implicative
coinductive lemmata are used in proofs alongside all other Horn clauses
in~\cite{FuKSP16} indeed resembles a program transformation method
when considered from the logic programming point of view.
In reality, however, different fragments of the
calculus given in~\cite{FuKSP16} allow proofs for Horn formulae which, when added
to the initial program, may lead to inductively or coinductively  unsound
extensions.  We analyse this situation carefully, throughout the
technical sections that follow.  In this way, we highlight which program transformation methods can
be soundly borrowed from existing work on corecursive resolution.
We will use the formulation of corecursive type class resolution
given by Fu \etal~\cite{FuKSP16}.  This extends Howard's simply-typed lambda
calculus~\cite{Howard80,FuK16} with
a resolution rule and a $\nu$-rule.  The resulting
calculus is general and  accounts for all previously suggested kinds of type
class resolution.


\knote{
    A question to Franta, only for times when all other urgent issues are
    resolved (may be even for after submission time): it looks somewhat
    plausible that corecursive type class resolution is not just sound, but also
    complete relative to greatest Herbrand models. Can we check this at some
    point?  It would be cool to have that, as coinductive completeness rarely
    hods...
}

\section{Preliminaries}
\label{sec:prel}

This section describes our notation and defines the models that we will use in the rest of
the paper.
As is  standard, a first-order signature $\Sigma$  consists of the set
$\mathcal{F}$ of function symbols and the set $\mathcal{P}$ of predicate
symbols, all of which possess an \emph{arity}.
Constants are function
symbols of arity $0$.  We also assume a countable set $\mathcal{V}$ of
variables. 
Given $\Sigma$ and $\mathcal{V}$, we have the following standard definitions:

\begin{definition}[Syntax of Horn formuale and logic programs]\label{df:syntax}
    \begin{alignat*}{3}
        \text{First-order term} && \quad \text{Term} & ::= \mathcal{V} \mid \mathcal{F}(Term, \dots, Term)\\
        \text{Atomic formula} && \text{At} & ::= \mathcal{P}(Term, \dots, Term)\\
        \text{Horn formula (clause)} && \text{CH} & ::= \text{At}, \dots,
        \text{At} \Rightarrow \text{At}\\
        \text{Logic program} && \text{Prog} & ::=  \text{CH}, \ldots , \text{CH}
    \end{alignat*}
\end{definition}

\noindent
We use identifiers $t$ and $u$ to denote terms and 
$A, B, C$ to
denote  atomic formulae.  We use $P$ with indicies to refer to elements of \emph{Prog}.
We say that a term or an atomic formula is \emph{ground} if it contains no variables.
We assume that all variables in Horn formulae 
are implicitly universally quantified. Moreover, restriction
\ref{r:noexist} from Section~\ref{sec:intro} requires that there are no
\emph{existential variables}, \ie given a clause $B_1, \ldots , B_n \Rightarrow
A$,  if a variable occurs in $B_i$, then it also occurs in $A$. 
We use the common term \emph{formula}  to refer to both atomic formulae and to Horn
formulae. A \emph{substitution} and the \emph{application} of a substitution to
a term or a formula are defined in the usual way.  We denote application of a
substitution $\sigma$ to a term $t$ or to an atomic formula $A$ by $\sigma t$
and $\sigma A$ respectively. We denote composition of substitutions $\sigma$ and
$\tau$ by $\sigma \circ \tau$.  A substitution $\sigma$ is a \emph{grounding}
substitution for a term $t$ if $\sigma t$ is a ground term, and similarly for an
atomic formula.

\subsection{Models of Logic Programs}
\label{sec:mod}

Throughout this paper, we use the standard definitions of the least and greatest
Herbrand models.
Given a signature $\Sigma$, the \emph{Herbrand universe} $\US$ is
the set of all ground terms over $\Sigma$.  
Given a Herbrand universe $\US$ we define the \emph{Herbrand base} $\BS$  as the 
set of all atoms consisting only of ground terms in $\US$.

\begin{definition}[Semantic operator] 
    \label{def:semop}
    Let $P$ be a logic program over signature $\Sigma$. The mapping $\TP : 2^{\BS}
    \to 2^{\BS}$ is defined as follows. Let $I$ be a subset of $\BS$.
    \begin{align*}
        \TP(I) = \{ A \in \BS \mid &~ B_1, \dots B_n \Rightarrow A \text{~is a ground
            instance of a clause in $P$,}\\
            & ~\text{and} ~ \{B_1, \dots, B_n\} \subseteq I \}
    \end{align*}
\end{definition}
The 
operator gives inductive and coinductive interpretation to a logic
program.
\begin{definition}
    Let $P$ be a logic program.

    \begin{itemize}
        \item The \emph{least Herbrand model} is the least set $\MP \in \BS$
            such that $\MP$ is a fixed point of $\TP$.
        \item The \emph{greatest Herbrand model} is the greatest set $\McP \in
            \BS$ such that $\McP$ is a fixed point of $\TP$.
    \end{itemize}
\end{definition}
Lloyd~\cite{Lloyd87} introduces the operators
$\downarrow$ and $\uparrow$  and proves that
$\TP \downarrow \omega$ gives the greatest Herbrand model of $P$, 
and that $\TP \uparrow \omega$ gives the least Herbrand model of $P$.
We will use these constructions in our own proofs.
The validity of a formula in a model is defined as usual. An atomic formula is
\emph{valid} in a model $I$ if and only if for any
grounding substitution $\sigma$, we have $\sigma F \in I$. A Horn formula $B_1, \dots, B_n
\Rightarrow A$ is valid in $I$ if for any substitution $\sigma$, if $\sigma B_1,
\dots, \sigma B_n$ are valid in $I$ then $\sigma A$ is valid in $I$.
We use the notation $P \vDash_{ind} F$ to denote that a formula $F$ is valid in
$\MP$ and $P \vDash_{coind} F$ to denote that a formula $F$ is valid in
$\McP$. 

\begin{restatable}{lemma}{lemsemprop}
    \label{lem:sem_prop}
    Let $P$ be a logic program and let $\sigma$ be a substitution.  The following
    holds:
    \begin{enumerate}[label=\alph*)]
        \item If $(~\Rightarrow A) \in P$ then both
                $P \vDash_{ind} \sigma A$ and
                $P \vDash_{coind} \sigma A$
        \item If, for all $i$, $P \vDash_{ind} \sigma B_i$
            and $(B_1, \dots, B_n \Rightarrow A) \in P$
            then $P~\vDash_{ind}~\sigma A$
        \item If, for all $i$, $P \vDash_{coind} \sigma B_i$
            and $(B_1, \dots, B_n \Rightarrow A) \in P$
            then $P \vDash_{coind} \sigma A$
    \end{enumerate}

\end{restatable}

\noindent The proof of this lemma and of the other results of this paper can be found in Appendix~\ref{sec:appendix}. 

\begin{fullproof}{proofsemprop}
    a) Let $P$ be a logic program such that $(~\Rightarrow A) \in P$.  By
    Definition \ref{def:semop} of the semantic operator, for any grounding
    substitution $\tau$, $\tau A \in \TP(\MP)$.  Since $\MP$ is a fixed point of
    $\TP$ also $\tau A \in \MP$ and by definition of validity of a formula, $P
    \vDash_{ind} A$ and also, for any substitution $\sigma$, $P \vDash_{ind}
    \sigma A$. Since we do not use the fact that $\MP$ is the least fixed point the
    proof for coinductive case is identical.

    b) Let be $P$, $A$, $B_1$, \dots, $B_n$ as above. Assume, for all $i$,
    $P \vDash_{ind} B_i$ whence, for all $i$, for any grounding substitution $\sigma$,
    $\sigma B_i \in \MP$. By Definition \ref{def:semop} of semantic
    operator, $\sigma A \in \TP(\MP)$. Since
    $\MP$ is a fixed point also $\sigma A \in
    \MP$ and $P \vDash_{ind} \sigma A$.

    c) Note that the proof of b) does not make any use of the fact that
    $\MP$ is the least fixed point. Therefore use the proofs of 
    b) {\it mutatis mutandis}. 
    \qed
\end{fullproof}

\subsection{Proof Relevant Resolution}
\label{sec:res}

In~\cite{FuKSP16}, the usual syntax of Horn  formulae was embedded into a
type-theoretic framework, with Horn formulae seen as types inhabited by
proof terms.  In this setting, a judgement has the form $\Phi \vdash e: F$,
where $e$ is a proof term inhabiting formula $F$, and $\Phi$ is an \emph{axiom
environment} containing annotated Horn formulae that correspond to the given logic
program.
This gives rise to the following syntax, in addition to
that of Definition~\ref{df:syntax}.  We assume a set of proof term symbols $K$, and  a
set of proof term variables $U$.

\begin{definition}[Syntax of proof terms and axiom environments]
    \begin{alignat*}{3}
        \text{Proof term} && \quad E & ::= K \mid U \mid  E~E
            \mid \lambda U . E \mid \nu U . E\\
            \text{Axiom environment} && \quad Ax & ::= \cdot \mid Ax,(E : CH)
    \end{alignat*}

\end{definition}

\noindent We use the notation $\kappa$ with indices to refer to elements of $K$,
$\alpha$ and $\beta$ with indices to refer to elements of $U$, $e$ to refer to
proof terms in $E$, and $\Phi$ to refer to axiom environments in \emph{Ax}.
Given a judgement  $\Phi \vdash e: F$, we call $F$ an \emph{axiom} if $e \in
K$, and we call $F$ a \emph{lemma} if $e \notin K$ is a closed term, \ie
it contains no free variables.  A proof term $e$ is in
\emph{head normal form} (denoted $\HNF(e)$), if $e = \lambda \underline{\alpha}.\kappa~\underline{e}$ where
$\underline{\alpha}$ and $\underline{e}$ denote (possibly empty) sequences of
variables $\alpha_1, \dots, \alpha_n$ and proof terms $e_1 \dots e_m$
respectively where $n$ and $m$ are known from the context or are unimportant.
The intention of the above definition is to interpret logic programs, seen as
sets of Horn formulae, as types.  Example~\ref{ex:pair} shows how
the proof term symbols $\kappa_1$ and $\kappa_2$ can be used to annotate 
clauses in the given logic program. 
We capture this intuition in the following formal definition:

\begin{definition}
  \label{def:axenv}
  Given a logic program $P_{A}$ consisting of Horn clauses $H_1, \ldots, H_n$, with
  each $H_i$ having the shape $B_1^i, \ldots , B_k^i \Rightarrow A^i$, the axiom
  environment $\Phi_A$ is defined as follows.  We assume proof term symbols
  $\kappa_1, \ldots , \kappa_n$, and define, for each $H_i$, $\kappa_i : B_1^i,
  \ldots , B_k^i \Rightarrow A^i$.
\end{definition}

\noindent Revisiting Example~\ref{ex:pair}, we can say that it shows the result
of translation of the program $P_{Pair}$ into $\Phi_{Pair}$ and $\Phi_{Pair}$
is an axiom environment for the logic program $P_{Pair}$. In general, we say
that $\Phi_A$ is an axiom environment for a logic program $P_A$ if and only if there 
is a translation of $P_A$ into $\Phi_A$. We drop the index $A$ where 
it is known or unimportant.
Restriction \ref{r:nonoverlap} from Section~\ref{sec:intro} requires that
axioms in an axiom environment do not overlap. 
However, a lemma may overlap with other
axioms and lemmata---only axioms are subject to restriction
\ref{r:nonoverlap}.
We refer the reader to~\cite{FuKSP16} for complete exposition of proof-relevant resolution.
In the following sections, we will use this syntax to gradually introduce inference rules for proof-relevant corecursive resolution. 
We start with its ``inductive'' fragment, \ie the fragment that is sound relative to the least Herbrand models, and then in subsequent sections consider
its two coinductive extensions (which are both sound with respect to the greatest Herbrand models).

\section{Inductive Soundness of Type Class Resolution}
\label{sec:ind_sound}

This section describes the inductive fragment of the calculus for the
extended type class resolution that was introduced by Fu \etal~\cite{FuKSP16}.  We
reconstruct the standard theorem of universal inductive soundness for the
resolution rule.
We consider an extended version of type class resolution,  working with queries given by Horn formulae, rather than just atomic formulae.
We show that the resulting proof system is inductively sound, but coinductively unsound; we also show that it is incomplete.
Based on these results, we discuss the program transformation methods that can arise.


\begin{definition}[Type class resolution]
\begin{align}
    \inference[if $(e : B_1,\dots, B_n \Rightarrow A) \in \Phi$]
    {
        \Phi \vdash e_1 : \sigma B_1 \quad \cdots 
        \quad \Phi \vdash e_n : \sigma B_n 
    } 
    {
        \Phi
        \vdash e~e_1 \cdots e_n : \sigma A
    } \tag{\sc Lp-m} \label{ir:lp-m}
\end{align}
\end{definition}
If, for a given atomic formula $A$, and a given environment $\Phi$,  $\Phi
\vdash e : A$ is derived using the  \ref{ir:lp-m} rule we say that $A$ is entailed
by $\Phi$  and that the proof term~$e$ witnesses this entailment.  We define
derivations and derivation trees resulting from applications of the above rule
in the standard way (\cf Fu \etal~\cite{FuKSP16}).

\begin{example}
    Recall the logic program $P_{Pair}$ in Example~\ref{ex:pair}.
    The inference steps for 
    $\mathtt{eq}(\mathtt{pair}(\mathtt{int},~\mathtt{int}))$
    correspond to the following derivation tree:
\begin{prooftree}
    \footnotesize
    \AxiomC{}
    \UnaryInfC{$
        \Phi_{Pair}
        \vdash \kappa_2 : \mathtt{eq}(\mathtt{int})
    $}
    \AxiomC{}
    \UnaryInfC{$
        \Phi_{Pair}
        \vdash \kappa_2 : \mathtt{eq}(\mathtt{int})
    $}
    \BinaryInfC{$
        \Phi_{Pair}
        \vdash
        \kappa_1 \kappa_2 \kappa_2 :
        \mathtt{eq}(\mathtt{pair}(\mathtt{int},\mathtt{int}))
    $}
\end{prooftree}
\end{example}
The above entailment is inductively sound, \ie it is sound with respect to the
least Herbrand model of $P_{Pair}$:

\begin{restatable}{theorem}{thmindsound}
    \label{thm:ind_sound}
    Let $\Phi$ be an axiom
    environment for a logic program $P$, and let $\Phi \vdash e : A$ hold.
    Then $P \vDash_{ind} A$.
\end{restatable}
\begin{shortproof}
    By structural induction on the derivation tree and construction of the least
    Herbrand model, using Lemma~\ref{lem:sem_prop}.
    \qed
\end{shortproof}
\begin{fullproof}{proofindsound}
    By structural induction on the derivation tree.

    \emph{Base case}: Let the derivation be $\inference[]{}{\Phi \vdash \kappa :
    A}$ for an atomic formula $A'$, a proof term symbol $\kappa$,
    and a substitution $\sigma$ such that $A = \sigma A'$. By definition of the
    rule \ref{ir:lp-m} there is a clause  $(\kappa : ~\Rightarrow A') \in \Phi$
    and from the Definition \ref{def:axenv} of axiom environment also $~ \Rightarrow
    A' \in P$. From the Lemma~\ref{lem:sem_prop} part a) follows that $P
    \vDash_{ind} \sigma A'$.

    \emph{Inductive case}: Let the last step in the derivation tree be\\
    $\inference[]{\Phi \vdash e_1 : \sigma B_1 \dots \Phi \vdash e_n : \sigma
    B_n}{\Phi \vdash \kappa e_1 \dots e_n : \sigma A'}$ for atomic formulae
    $A'$, $B_1$, \dots, $B_n$, a proof term symbol $\kappa$, a substitution
    $\sigma$ and proof witnesses $e_1$, \dots, $e_n$ such that $A = \sigma A'$.
    From the definition of the rule
    \ref{ir:lp-m} there is a clause $(\kappa : B_1, \dots B_n \Rightarrow A') \in
    \Phi$ and from the Definition \ref{def:axenv} of axiom environment also $B_1,
    \dots, B_n \Rightarrow A' \in P$. 
    From the induction assumption, for $i \in \{1, \dots, n\}$, $P \vDash_{ind}
    \sigma B_i$ and by the Lemma~\ref{lem:sem_prop} part b), $P \vDash_{ind}
    \sigma A'$.
    \qed
\end{fullproof}

\noindent
The rule \ref{ir:lp-m} also plays a crucial role in the coinductive fragment of
type class resolution, as will be discussed in Sections \ref{sec:coind_sound}
and \ref{sec:ext_coind_sound}. 
We now discuss the other rule that is present in~the work of Fu
\etal~\cite{FuKSP16}, i.e the rule that allows us to prove Horn formulae: 

\begin{definition}
\begin{gather*}
    \inference[]
    {
        \Phi, (\beta_1 : ~\Rightarrow B_1), \dots, (\beta_n : ~\Rightarrow B_n)
            \vdash  e : A 
    }
    {
        \Phi \vdash \lambda \beta_1,\dots,\beta_n . e : B_1,\dots,B_n
        \Rightarrow A 
    } \tag{\sc Lam} \label{ir:lam}
\end{gather*}
\end{definition}

\begin{example} 
\noindent To illustrate the use of the \ref{ir:lam} rule,
consider the following program:
Let $P$ consist of two clauses: $A \Rightarrow
B$ and $B \Rightarrow C$. Both the least and the greatest Herbrand model of $P$
are empty. Equally, no 
formulae can be derived from the
corresponding axiom environment
by the \ref{ir:lp-m} rule.  However, we can derive $A \Rightarrow C$
by using a combination of the \ref{ir:lam} and \ref{ir:lp-m} rules.
    Let $\Phi = (\kappa_1 : A \Rightarrow B), (\kappa_2 : B \Rightarrow C)$. The
    following is then a derivation tree for a formula $A \Rightarrow C$:

    \begin{prooftree}
        \footnotesize
        \AxiomC{}
        \UnaryInfC{$
            \Phi, (\alpha : ~\Rightarrow A) \vdash \alpha : A
        $}
        \UnaryInfC{$
            \Phi, (\alpha : ~\Rightarrow A) \vdash \kappa_1 \alpha : B
        $}
        \UnaryInfC{$
            \Phi, (\alpha : ~\Rightarrow A) \vdash \kappa_2 (\kappa_1 \alpha) : C
        $}
        \RightLabel{\sc Lam}
        \UnaryInfC{$
            \Phi \vdash \lambda \alpha . \kappa_2 (\kappa_1 \alpha) 
                : A \Rightarrow C
        $}
    \end{prooftree}
    When there is no label on the right-hand side of an inference step, inference
    uses the  \ref{ir:lp-m} rule. We follow this convention throughout the
    paper.
\end{example}

We can show that the calculus comprising the rules  \ref{ir:lp-m} and
\ref{ir:lam} is again (universally) inductively sound.

\begin{restatable}{lemma}{lemsteplamind}
    \label{lem:step_lam_ind}
    Let $P$ be a logic program and let $A$, $B_1$, \dots, $B_n$ be atomic formulae.
    If $P, (~\Rightarrow B_1), \dots, (~\Rightarrow B_n)
    \vDash_{ind} A$ then $P \vDash_{ind} B_1, \dots, B_n \Rightarrow A$.
\end{restatable}
\begin{shortproof}
    By induction on construction of $\MP$. \qed
\end{shortproof}
\begin{fullproof}{proofsteplamind}
    Assume that $P, (~\Rightarrow B_1), \dots, (~\Rightarrow B_n)
    \vDash_{ind} A$. From the Definition \ref{def:semop} there is the least $n$
    such that for any grounding substitution $\tau$, $(\tau \circ \sigma) A \in
    \\
    \TP{}_{,(~\Rightarrow B_1), \dots, (~\Rightarrow B_n)}~\uparrow~n$. Consider
    any substitution $\sigma$ and suppose that for all $i$, $P \vDash \sigma
    B_i$.  From the definition of validity for any grounding $\tau$ for all
    $i$,
    $(\tau \circ \sigma) B_i \in \MP$ hence there is the least $m$ such that $(\tau
    \circ \sigma) B_i \in \TP \uparrow m$. From the assumption, for any grounding 
    substitution $\tau$ also $(\tau \circ \sigma) A \in \TP \uparrow (n + m)$
    and $P \vDash \sigma A$. Hence $P \vDash_{ind} B_1, \dots, B_n \Rightarrow
    A$.  \qed
\end{fullproof}
\vskip-1em
\begin{restatable}{theorem}{thmlpmlam}
    \label{thm:lp-m_lam}
  Let $\Phi$ be an axiom environment for a logic program $P$ and $F$ a formula.  
  Let $\Phi \vdash e : F$ be by the \ref{ir:lp-m} and \ref{ir:lam} rules.
  Then $P \vDash_{ind} F$.
\end{restatable}
\begin{shortproof}
    By structural induction on the derivation tree using Lemmata
    \ref{lem:sem_prop} \& \ref{lem:step_lam_ind}.
    \qed
\end{shortproof}
\begin{fullproof}{prooflpmlam}
    By structural induction on the derivation tree.

    \emph{Base case}: Let the derivation be $\inference[]{}{\Phi \vdash \kappa :
    A}$ for an atomic formula $A'$, a constant symbol $\kappa$,
    and a substitution $\sigma$ such that $A = \sigma A'$. By definition of the
    rule \ref{ir:lp-m} there is a clause  $(\kappa : ~\Rightarrow A') \in \Phi$
    and from the Definition \ref{def:axenv} of axiom environment also $~ \Rightarrow
    A' \in P$. From the Lemma~\ref{lem:sem_prop} part a) follows that
    $P \vDash_{ind} \sigma A'$.

    \emph{Inductive case}: Let the last step in the derivation tree be by the
    rule \ref{ir:lp-m} thus of the form
    $\inference[]{\Phi \vdash e_1 : \sigma B_1 \dots \Phi \vdash e_n : \sigma
    B_n}{\Phi \vdash \kappa e_1 \dots e_n : \sigma A'}$ for atomic formulae 
    $A$, $B_1$, \dots, $B_n$, a proof term symbol $\kappa$, a substitution
    $\sigma$ and proof witnesses $e_1$, \dots, $e_n$.
    From the definition of the rule 
    \ref{ir:lp-m} there is a clause $(\kappa : B_1, \dots B_n \Rightarrow A) \in
    \Phi$ and from the Definition \ref{def:axenv} of axiom environment also $B_1,
    \dots, B_n \Rightarrow A \in P$. 
    From the induction assumption, for $i \in \{1, \dots, n\}$, $\Phi \vDash_{ind}
    \sigma B_i$ and by the Lemma~\ref{lem:sem_prop} part b), $P \vDash_{ind} 
    \sigma A'$.

    Let the last step of the derivation be by the rule \ref{ir:lam} thus of the
    form
    $\inference[]{\Phi, (\beta_1 : ~\Rightarrow B_1), \dots, (\beta_n :
    ~\Rightarrow B_n) \vdash e : A}{\Phi \vdash \lambda \beta_1, \dots, \beta_n
    . e : B_1, \dots, B_n \Rightarrow A}$ for atomic formulae $A$, $B_1$, \dots,
    $B_n$, proof term $e$, and variables $b_1$, \dots,
    $b_n$.
    From the induction assumption, $P, (~\Rightarrow B_1), \dots, ( ~\Rightarrow
    B_n) \vDash A$ and from the Lemma~\ref{lem:step_lam_ind} also $P \vDash_{ind}
    A$.
    \qed
\end{fullproof}
\vskip-3em~
\subsubsection{Inductive Completeness and Incompleteness of the Proof System  \ref{ir:lp-m} + \ref{ir:lam}.}
In principle, one can consider two
different variants. 
Extending the standard results of~\cite{Lloyd87}, our first formulation is:

\vspace{6pt}
\noindent
\textbf{Inductive Completeness-1:} 
\emph{if a ground atomic formula $F$
  is in $\MP$, then  $\Phi_P \vdash e:F$ is in the  \ref{ir:lp-m} + \ref{ir:lam} proof system.}

\vspace{6pt}
\noindent
  Such a result can be proved, as in~\cite{Lloyd87}, by
  straightforward induction on the construction of $\MP$. Such a proof will be based solely on the properties of the rule \ref{ir:lp-m} and on the properties of the semantic operator $\TP$ that is used to construct
    the least Herbrand
    models.

\vspace{12pt}
\noindent
An alternative formulation of the completeness result, this time involving implicative formulae and hence the rule  \ref{ir:lam} in the proof, would be:

\vspace{6pt}
\noindent
\textbf{Inductive Completeness-2:} 
\emph{if  $\MP \vDash_{ind} F$ then  $\Phi_P \vdash e:F$ is in the  \ref{ir:lp-m} + \ref{ir:lam} proof system.}
\vspace{6pt}

\noindent
However, this result would not hold for either system \ref{ir:lp-m} or   \ref{ir:lp-m} + \ref{ir:lam}. Consider the following examples.
\newknote{check the above please?}

\begin{example}
    \label{ex:incomp1}
    Let $\Sigma$ be a signature consisting of a unary predicate symbol $\mathtt{A}$, a
    unary function symbol $\mathtt{f}$ and a constant function symbol $\mathtt{g}$. 
    Let $P_6$ be a program given by the following axiom environment:
\begin{alignat*}{3}
    \kappa_1 :~ & && \Rightarrow \mathtt{A}(\mathtt{f}(x))\\
    \kappa_2 :~ & && \Rightarrow \mathtt{A}(\mathtt{g})
\end{alignat*}
\noindent
The least Herbrand model of $P_6$ is $\MP{_{{}_6}} = \{A(g), A(f(g)), A(f(f(g))), \dots
\}$. Therefore, $P \vDash_{ind} \mathtt{A}(x)$. However, neither
$\kappa_1$ nor $\kappa_2$ matches $\mathtt{A}(x)$ and there is thus no way
to construct a  proof term $e$ satisfying:
    \begin{prooftree}
\vskip-1.5em
        \footnotesize
        \AxiomC{$\cdots$}
        \RightLabel{\sc Lp-m}
        \UnaryInfC{$
            P \vdash
            e : \mathtt{A}(x)
        $}
    \end{prooftree}
\end{example}
\noindent
We demonstrate the incompleteness of the proof system \ref{ir:lp-m} + \ref{ir:lam} through
the following example:

\begin{example}
    \label{ex:incomp2}
    Let $\Sigma$ be a signature consisting of the unary predicate symbols
    $\mathtt{A}$ and $\mathtt{B}$, and a constant function symbol $\mathtt{f}$.
    Consider a program $P_7$ given by the following axiom environment:
    \vskip-2.2em
\begin{alignat*}{3}
    \kappa_1 :~ & && \Rightarrow \mathtt{A}(\mathtt{f})\\
    \kappa_2 :~ & && \Rightarrow \mathtt{B}(\mathtt{f})
\end{alignat*}
\noindent
The least Herbrand model is $\MP{_{{}_7}} = \{A(f), B(f) \}$. Therefore 
$P \vDash_{ind} \mathtt{B}(x) => \mathtt{A}(x)$. 
However, any proof of $\mathtt{B}(x) => \mathtt{A}(x)$
needs to show that:
    \begin{prooftree}
\vskip-1.5em
        \footnotesize
        \AxiomC{$\cdots$}
        \UnaryInfC{$
            (P, \alpha : ~\Rightarrow \mathtt{B}(x)) \vdash
            e : \mathtt{A}(x)$}
        \RightLabel{\sc Lam}
        \UnaryInfC{$
            P \vdash
            \lambda \alpha . e : \mathtt{B}(x) \Rightarrow \mathtt{A}(x)
        $}
    \end{prooftree}
\noindent
where $e$ is a proof term. This proof will not succeed since no axiom or hypothesis matches $\mathtt{A}(x)$.
\end{example}

\vskip-3em~
\subsubsection{Related Program Transformation Methods}
For Fu \etal~\cite{FuKSP16}, the main purpose of introducing the rule~\ref{ir:lam} 
was to increase expressivity of the proof system.
In particular, obtaining an entailment $\Phi \vdash e : H$ of a Horn
formula $H$ enabled 
the environment $\Phi$ to be extended with $e:H$, which could be used in 
future proofs.
We show that transforming (the standard, untyped) logic programs in this way is
inductively sound.
The following theorem follows from Lemma~\ref{lem:step_lam_ind}:

\begin{restatable}{theorem}{thmindtrans}
    \label{thm:lp-m_lam2}
  Let $\Phi$ be an axiom environment for a logic program $P$, and let  $\Phi \vdash e : F$ for a 
  formula $F$ by  the \ref{ir:lp-m} and \ref{ir:lam} rules.  Given a formula $F'$,
  $P \vDash_{ind} F'$ iff  $P, F \vDash_{ind} F'$.
\end{restatable}
\begin{fullproof}{proofindtrans}
    By the Theorem \ref{thm:ind_sound}, $P \vDash_{ind} F$. Therefore, $\MP$ is
    a model of $F$ and $\MP = \MP {}_{, F}$. Hence $P \vDash_{ind} F'$ iff  $P,
    F \vDash_{ind} F'$.
    \qed
\end{fullproof}
\noindent
Note, however,  that the above theorem is not as trivial as it looks, in
particular, it would not hold coinductively, \ie if we changed $ \vDash_{ind}$
to  $ \vDash_{coind}$ in the statement above.  Consider the following proof of
the formula $A \Rightarrow A$:
   
\begin{example}
    \label{ex:a_impl_a}
Using the ~\ref{ir:lam} rule, one can prove $\emptyset \vdash
\lambda\alpha . \alpha : A => A$:
    \begin{prooftree}
        \footnotesize
        \AxiomC{}
        \UnaryInfC{$(\alpha : ~\Rightarrow A) \vdash \alpha : A$}
        \RightLabel{\sc Lam}
        \UnaryInfC{$\emptyset \vdash \lambda \alpha . \alpha : A \Rightarrow A$}
    \end{prooftree}
    %
    \newknote{I commented out one line above, as it confused me in the context of this example. it was about inductive models, whereas the example talks about coinductive models}
Assume a program consisting of a single formula $A \Rightarrow B$.
Both the least and the greatest Herbrand model of this program are empty.
However, adding the formula $A \Rightarrow A$ to the program results in the
greatest Herbrand model $\{A, B\}$. Thus, 
$\McP \neq \mathcal{M}'_{P, (A \Rightarrow A)}$.
\end{example}


\section{Coinductive Soundness of Corecursive Type Class Resolution}
\label{sec:coind_sound}

The \ref{ir:lp-m} rule may result in
non-terminating resolution.  This can be demonstrated by the program
$P_{EvenOdd}$ and the query
$\mathtt{?}~\mathtt{eq}(\mathtt{evenList}(\mathtt{Int}))$ from Section
\ref{sec:intro}. L\"ammel and Peyton Jones observed~\cite{LammelJ05} that in
such cases there may be a cycle in the inference that can be detected. 
This treatment of cycles amounts to coinductive reasoning and results in
building a corecursive proof witness---\ie a \emph{(co-)recursive dictionary}.

\begin{definition}[Coinductive type class resolution]
\begin{align*}
    \inference[if $\HNF(e)$]
    {
        \Phi, (\alpha : ~\Rightarrow A) \vdash e : A  
    } 
    {
        \Phi \vdash \nu \alpha . e : A
    } \tag{\sc Nu'} \label{ir:nu'}
\end{align*}
\end{definition}
The side condition of \ref{ir:nu'} requires the proof witness to be in
head normal form.
Since, in this section, we are working with a calculus consisting of the rules
\ref{ir:lp-m} and \ref{ir:nu'}, there is no way to introduce 
a $\lambda$-abstraction into a proof witness.
Therefore, in this section, we restrict
ourselves to head normal form terms of the form $\kappa~\underline{e}$.

\begin{example}
    Recall the program $P_{EvenOdd}$ in Example~\ref{ex:oddeven}. The
    originally non-terminating resolution trace for the  query
    $\mathtt{?}~\mathtt{eq}(\mathtt{evenList}(\mathtt{int}))$ is resolved using the
    \ref{ir:nu'} rule as follows:
    \begin{prooftree}
        \footnotesize
    \vskip-2em
        \AxiomC{}
        \UnaryInfC{$
                \kappa_3 : \mathtt{eq}(\mathtt{int})
        $}
        \noLine
        \UnaryInfC{$
                \quad\quad
                \vdash \kappa_3 : \mathtt{eq}(\mathtt{int})
        $}

        \AxiomC{}
        \UnaryInfC{$
                \kappa_3 : \mathtt{eq}(\mathtt{int})
        $}
        \noLine
        \UnaryInfC{$
                \quad\quad\quad
                \vdash \kappa_3 : \mathtt{eq}(\mathtt{int})
        $}
        \AxiomC{}
        \UnaryInfC{$
                \alpha : ~\Rightarrow \mathtt{eq}(\mathtt{evenList}(\mathtt{int}))
        $}
        \noLine
        \UnaryInfC{$
                \quad\quad
                    \vdash \alpha : \mathtt{eq}(\mathtt{evenList(\mathtt{int})})
        $}
        \BinaryInfC{$
            \Phi_{EvenOdd}, 
                \alpha : \_
                \vdash \kappa_1 \kappa_3 \alpha : 
                \mathtt{eq}(\mathtt{oddList}(\mathtt{int}))
        $}
        \BinaryInfC{$
            \Phi_{EvenOdd}, \alpha : \_
            \vdash \kappa_2 \kappa_3 (\kappa_1 \kappa_3 \alpha) : 
                \mathtt{eq}(\mathtt{evenList}(\mathtt{int}))
        $}
        \RightLabel{\sc Nu'}
        \UnaryInfC{$
            \Phi_{EvenOdd} 
            \vdash \nu \alpha . \kappa_2 \kappa_3 (\kappa_1 \kappa_3 \alpha) : 
                \mathtt{eq}(\mathtt{evenList}(\mathtt{int}))
        $}
    \end{prooftree}
    Note that we abbreviate repeated formulae in the environment using an underscore. We will use this notation in the rest of the
    paper.
\end{example}

We can now discuss the coinductive soundness of the
\ref{ir:nu'} rule, \ie its soundness relative to the greatest Herbrand models.  We
note that, not surprisingly (\cf \cite{Sangiorgi09}),  the  \ref{ir:nu'} rule is inductively unsound. Given a program
consisting of just one clause: $\kappa: A \Rightarrow A$, we are able to
use the rule  \ref{ir:nu'} to entail $A$ (the derivation of this will be similar to, albeit
a lot simpler than, that in the above example). However, $A$ is not in the least
Herbrand model of this program.
Similarly, the formula $ \mathtt{eq}(\mathtt{oddList}(\mathtt{int}))$ that was proved above is also not inductively sound.
%
Thus, the coinductive fragment of the extended corecursive resolution is only
coinductively sound.  When proving the coinductive soundness of  the
\ref{ir:nu'}  rule, we must carefully choose the proof method by which we proceed.
Inductive soundness of the  \ref{ir:lp-m} rule was proven by induction on the
derivation tree and through the construction of the least Herbrand models by iterations
of $\TP$.  Here, we give an  analogous result, where coinductive soundness is
proved by structural coinduction on the iterations of the semantic operator
$\TP$.

In order for the principle of structural coinduction to be applicable in our
proof, we must ensure that the construction of the greatest Herbrand model is
completed within $\omega$ steps of iteration of $\TP$.  This does not hold in
general for the greatest Herbrand model construction, as was shown e.g.
in~\cite{Lloyd87}. However, it does hold for the restricted shape of Horn
clauses we are working with.  It was noticed by Lloyd~\cite{Lloyd87} that
Restriction \ref{r:noexist} from Section~\ref{sec:intro} implies that the $\TP$ operator 
converges in at most $\omega$ steps. We will capitalise on this fact.
The essence of the coinductive soundness of \ref{ir:nu'} is captured by the
following lemma:

\begin{restatable}{lemma}{lemstepnuprm}
    \label{lem:step_nu'}
    Let $P$ be a logic program, let $\sigma$ be a substitution,
    and let $A$, $B_1$, \dots, $B_n$ be atomic
    formulae. If, $\forall {i \in \{1, \dots, n\}}$, $P, (~\Rightarrow
    \sigma A) \vDash_{coind} \sigma B_i$ and $(B_1, \dots, B_n \Rightarrow A)
    \in P$ then $P \vDash_{coind} \sigma A$.
\end{restatable}
\noindent
The proof of the lemma is similar to the proof of the Lemma \ref{lem:step_nu}
in the next section and we do not state it here.%
\begin{fullproof}{proofstepnuprm}
  Consider now the construction of the greatest Herbrand model for the program
  $P$ and proceed by coinduction with coinductive hypothesis: for all $n$,
  for any grounding substitution $\tau$, $(\tau \circ \sigma) A \in \TP
  \downarrow n$. Since, for any $\tau$, $(\tau \circ \sigma) A \in \TP
  \downarrow n$ the set $\TP \downarrow n$ is by definition of the operator
  $\TP$ the same as the set $\TP{}_{,(~\Rightarrow \sigma A)}$ and from the
  assumptions of the lemma and monotonicity of $\TP$ also, for all $i$, for any
  grounding substitution $\tau$, $(\tau \circ \sigma) B_i \in \TP \downarrow n$.
  Since $B_1, \dots, B_n => A \in P$ also $(\tau \circ \sigma) A \in \TP
  \downarrow (n + 1)$. We now apply the
  coinduction hypotheses to conclude that the same will be true for all
  subsequent iterations of $\TP$.  But then all instances of $\sigma A$ will be in
  the greatest Herbrand model of this program. Hence $P \vDash_{coind} \sigma A$
  \qed
\end{fullproof}%

\newknote{I cannot make any sense of the new proof by coinduction. It looks like someone was trying to cheat, and gave a mock proof instead of a real proof. I cannot work out: when Coinductive assumption is formed, when it is applied, and why should I know that `` from the
  assumptions of the lemma'' the assumption will be applicable. Overall, I cannot regard this as a proof, i am sorry. And this is the main result of the paper. }
  \fnote{I believe the proof is along the lines of the proof on the whiteboard
  modulo some technical details, which were necessary}

\vspace{1em}
\noindent
Finally, Theorem \ref{thm:lp-m_nu'} states universal coinductive soundness of
the coinductive type class resolution:

\begin{restatable}{theorem}{thmlpmnuprm}
    \label{thm:lp-m_nu'}
    Let $\Phi$ be an axiom environment for a logic program $P$ and $F$ a
    formula. Let $\Phi \vdash e : F$ be by the \ref{ir:lp-m} and \ref{ir:nu'}
    rules.  Then $\Phi \vDash_{coind} F$.
\end{restatable}
\begin{shortproof}
    By structural induction on the derivation tree using Lemmata
    \ref{lem:sem_prop} \& \ref{lem:step_nu'}.
    \qed
\end{shortproof}
\begin{fullproof}{prooflpmnuprm}
    By structural induction on the derivation tree.

    Base case: Let the derivation be with empty assumptions. Then it is by the
    rule \ref{ir:lam} and of the form $\inference[]{}{\Phi \vdash \kappa : \sigma
    A}$ for an atomic formula $A$, a constant symbol $\kappa$, and a
    substitution $\sigma$. By definition of the rule \ref{ir:lp-m} there is a
    clause  $(\kappa : ~\Rightarrow A) \in \Phi$.
    By Lemma~\ref{lem:sem_prop} c), $\Phi \vDash_{coind} \sigma A$.

    Inductive case: Let the last step be by the rule \ref{ir:lp-m} and of
    the form $\inference[]{\Phi \vdash e_1 : \sigma B_1
    \dots \Phi \vdash e_n : \sigma B_n}{\Phi \vdash \kappa e_1 \dots e_n :
    \sigma A}$
    for an atomic formulae $A$, $B_1$, \dots, $B_n$ a constant symbol $\kappa$,
    a substitution $\sigma$ and proof witnesses $e_1$, \dots, $e_n$.  By
    definition of the rule \ref{ir:lp-m} there is a clause $(\kappa : B_1,
    \dots B_n \Rightarrow A) \in \Phi$.
    
    By the induction assumption, for $i \in \{1, \dots, n\}$, $\Phi
    \vDash_{coind} B_i$ and by Lemma~\ref{lem:sem_prop} d), $\Phi
    \vDash_{coind} \sigma A$.

    Let the last step be by the rule \ref{ir:nu'} and of the form
    $\inference[]{\Phi, (\alpha : ~\Rightarrow A) \vdash e : A}{\Phi \vdash \nu
    \alpha . e : A}$ for an atomic formula $A$, a variable $\alpha$ and a proof
    witness $e$ in the head normal form. W.l.o.g. let $e = \kappa e_1 \dots e_n$.
    Therefore there is an 
    inference step of the
    form $\inference[]{\Phi \vdash e_1 : \sigma B_1'
    \dots \Phi \vdash e_n : \sigma B_n'}{\Phi \vdash \kappa e_1 \dots e_n :
    \sigma A'}$ for $\sigma A' = A$ and $(\kappa : B'_1, \dots B'_n \Rightarrow
    A') \in \Phi$. By the induction assumption, for all $i$, $\Phi, (\alpha :
    ~\Rightarrow A) \vDash B_i$.
    By Lemma~\ref{lem:step_nu'}, $\Phi \vDash_{coind} A$.
    \qed
\end{fullproof}


\subsubsection{Choice of Coinductive Models.}
Perhaps the most unusual feature of the semantics given in this section is the use of 
the greatest Herbrand models rather than
the greatest \emph{complete} Herbrand models. The latter is more common in the literature on coinduction in
logic programming \cite{Lloyd87,KomendantskayaJ15,SimonBMG07}.
 \emph{The greatest complete Herbrand models}
 are obtained as the
greatest fixed point of the semantic operator $\TP'$ on the \emph{complete
Herbrand base}, \ie the set of all finite and \emph{infinite} ground atomic
formulae formed by the signature of the given program.
This construction is preferred in the literature for two reasons.
Firstly, $\TP'$ reaches its greatest fixed point in at
most $\omega$ steps, whereas $\TP$ may take more than $\omega$ steps in the general case.
This is due to compactness
of the complete Herbrand base.  Moreover, greatest complete Herbrand models give
a more natural characterisation for programs like the one given by the clause
$\kappa_{inf}:  \mathtt{p}(x) \Rightarrow \mathtt{p}(\mathtt{f}(x))$.  The
greatest Herbrand model of that program 
is empty.  However, its greatest complete Herbrand model 
contains the infinite formula $\mathtt{p}(\mathtt{f}(\mathtt{f}(...))$. 
Restrictions \ref{r:nonoverlap} -- \ref{r:match}, imposed by type class
resolution, mean that the greatest Herbrand models regain those same advantages
as complete Herbrand models.
It was noticed by Lloyd~\cite{Lloyd87} that restriction \ref{r:noexist} implies
that the semantic operator converges in at most $\omega$ steps. Restrictions
\ref{r:nonoverlap} and \ref{r:match} imply that proofs by type class resolution
have a universal interpretation, \ie that they hold for all finite instances of
queries. Therefore, we never need to talk about programs for which only one
infinite instance of a query is valid.


\section{Coinductive Soundness of Extended Corecursive Type Class Resolution}
\label{sec:ext_coind_sound}

The class of problems that can be resolved by coinductive type class resolution
is limited to problems where a coinductive hypothesis is in atomic form.
Fu \etal \cite{FuKSP16} extended coinductive type
class resolution with implicative reasoning and adjusted the rule \ref{ir:nu'}
such that this restriction of coinductive type class resolution is relaxed:

\begin{definition}[Extended coinductive type class resolution]
\begin{align*}
    \inference[if $\HNF(e)$]
    {
        \Phi, (\alpha : B_1,\dots,B_n \Rightarrow A)
            \vdash e : B_1,\dots,B_n \Rightarrow A & 
    } 
    {
        \Phi 
            \vdash \nu \alpha . e
            : B_1,\dots,B_n \Rightarrow A
    } \tag{\sc Nu} \label{ir:nu}
    %
\end{align*}
\end{definition}
The side condition of the \ref{ir:nu}  rule requires the proof witness to be in
head normal form. However, unlike coinductive type class resolution, 
extended coinductive type class resolution  also uses the \ref{ir:lam}  rule and a
head normal term is of the form $\lambda \underline{\alpha} . \kappa
\underline{e}$ for possibly non-empty sequence of proof term variables $\underline{\alpha}$.
First, let us note that extended coinductive type class resolution indeed
extends the calculus of Section~\ref{sec:coind_sound}:

\begin{restatable}{proposition}{propadmisible}
    The inference rule \ref{ir:nu'} is admissible in the extended coinductive
    type class resolution.
\end{restatable}
\begin{fullproof}{proofadmisible}
    Let $\Phi$ be an environment, let $A$ be an atomic formula and let $\Phi,
    (\alpha : ~\Rightarrow A) \vdash e : A$ where $e$ is in $\HNF$.
    Then by the \ref{ir:lam} rule $\Phi, (\alpha : ~\Rightarrow A)\vdash \lambda
    \underline{\beta} . e : ~\Rightarrow A$ where $\underline{\beta}$ is an
    empty sequence of variables.  Therefore $\Phi, (\alpha : ~\Rightarrow
    A) \vdash e : ~\Rightarrow A$.  Since $e$ is in head normal
    form by the \ref{ir:nu} rule also $\Phi \vdash \nu \alpha . e : A$. 
    \qed
\end{fullproof}

\noindent
Furthermore, this is a proper extension.
The \ref{ir:nu}  rule allows queries to be entailed that were beyond the scope of
coinductive type class resolution. In Section \ref{sec:intro},
we demonstrated a derivation for query
$\mathtt{?}~\mathtt{eq}(\mathtt{bush}(\mathtt{int}))$ where no
cycles arise and thus the query cannot be resolved by coinductive type class
resolution.

\begin{example}
    Recall the program $P_{Bush}$ in Example~\ref{ex:bush}. The query
    $\mathtt{?}~\mathtt{eq}(\mathtt{bush}(\mathtt{int}))$ is resolved as follows:
    \begin{prooftree}
        \footnotesize
        \AxiomC{}
        \UnaryInfC{$
                \Phi_{Bush}
                \vdash
            $}
        \noLine
        \UnaryInfC{$
                \kappa_1 : \mathtt{eq}(\mathtt{int})
            $}
        \AxiomC{}
        \UnaryInfC{$
            (\beta : ~\Rightarrow \mathtt{eq}(x)) 
            $}
        \noLine
        \UnaryInfC{$
                \quad\quad
                \vdash 
                \beta : 
                \mathtt{eq}(x)
            $}
        \AxiomC{}
        \UnaryInfC{$
                (\beta : ~\Rightarrow \mathtt{eq}(x)) 
                \vdash \beta :
                \mathtt{eq}(x)
            $}
        \UnaryInfC{$
                (\alpha : \mathtt{eq}(x) \Rightarrow \mathtt{eq}(\mathtt{bush}(x))),
                (\beta : \_)
                \vdash
            $}
        \noLine
        \UnaryInfC{$
                \alpha \beta : \mathtt{eq}(\mathtt{bush}(x))
            $}
        \UnaryInfC{$
                (\alpha : \_),
                (\beta : \_)
                \vdash
                \alpha (\alpha \beta) :
                \mathtt{eq}(\mathtt{bush}(\mathtt{bush}(x)))
            $}
        \BinaryInfC{$
                \Phi_{Bush}, (\alpha : \_), (\beta : \_) 
                \vdash
                \kappa_2 \beta (\alpha (\alpha \beta)) :
                \mathtt{eq}(\mathtt{bush}(x))
            $}
        \RightLabel{\sc Lam}
        \UnaryInfC{$
            \Phi_{Bush}, (\alpha : \_) \vdash
                \lambda \beta . \kappa_2 \beta (\alpha (\alpha \beta)) :
                \mathtt{eq}(x) \Rightarrow \mathtt{eq}(\mathtt{bush}(x))
            $}
        \RightLabel{\sc Nu}
        \UnaryInfC{$
            \Phi_{Bush} 
                \vdash
                \nu \alpha . \lambda \beta . \kappa_2 \beta 
                    (\alpha (\alpha \beta)) : \mathtt{eq}(x) \Rightarrow
                    \mathtt{eq}(\mathtt{bush}(x))
                $}
        \BinaryInfC{$
            \Phi_{Bush} \vdash 
                (\nu \alpha . \lambda \beta . \kappa_2 \beta 
                (\alpha (\alpha \beta))) \kappa_1
                : \mathtt{eq}(\mathtt{bush}(\mathtt{int}))
        $}
    \end{prooftree}
\end{example}


Before proceeding with the proof of soundness of extended type class resolution
we need to show two intermediate lemmata. The first lemma states that inference
by the \ref{ir:nu} rule  preserves coinductive soundness:

\begin{restatable}{lemma}{lemstepnu}
    \label{lem:step_nu}
    Let $P$ be a logic program, let $\sigma$ be a substitution, and let $A$, $B_1$,
    \dots, $B_n$, $C_1$, \dots, $C_m$ be atomic formulae. If, for all $i$, %
    $P, B_1, \dots, 
    B_n, (B_1, \dots, B_n \Rightarrow \sigma A) \vDash_{coind} \sigma C_i$
    and $(C_1, \dots, C_m \Rightarrow A) \in P$ then 
    $P \vDash_{coind} B_1 \dots B_n \Rightarrow \sigma A$.
\end{restatable}
\begin{resproof}{proofstepnu}
    Consider the construction of the greatest Herbrand model of the
    program~$P$ and proceed by coinduction with coinductive hypothesis: for all
    $n$, $B_1, \dots, B_n => \sigma A$ is valid in $\TP \downarrow n$. Assume
    that, for a grounding substitution $\tau$, for all $i$, $\tau B_i \in \TP
    \downarrow n$. Then also $(\tau \circ \sigma) A \in \TP \downarrow n$. For the
    definition of the semantic operator, it follows from the monotonicity of the operator itself, and from
    the assumptions made by the lemma that $(\tau \circ \sigma) C_i \in \TP
    \downarrow n$.  Since $C_1, \dots, C_n => A \in P$ also $(\tau \circ \sigma)
    A \in \TP \downarrow (n + 1)$. If the assumption does not hold then from the
    monotonicity of $\TP$ it follows that, for all $i$, $\tau B_i \not\in \TP
    \downarrow (n + 1)$. Therefore, $B_1, \dots, B_n => \sigma A$ is valid if
    $\TP \downarrow (n + 1)$. We apply the
    coinductive hypothesis to conclude that the same holds for all
    subsequent iterations of $\TP$. Hence whenever, for a
    substitution $\tau$, all instances of $\tau B_1$ to $\tau B_n$ are in the
    greatest Herbrand model then also all instances of $(\tau \circ \sigma) A$
    are in the greatest Herbrand model. Hence $P \vDash_{coind} B_1, 
    \dots, B_n \Rightarrow A$.
    \qed
\end{resproof}

\newknote{Again, this is not a full proof. May be, after someone spends half a day, that person can reconstruct  a proof. But that does not mean that *you* gave a proof here.
We spent one day to recover the proof -- and I see no reflection of that work in the text. It remains as unreadable as it was when I first saw it.
}
  \fnote{I believe the proof is along the lines of the proof on the whiteboard
  modulo some technical details, which were necessary}

\noindent
The other lemma that we need in order to prove coinductive soundness of extended type
class resolution states that inference using
\ref{ir:lam} preserves coinductive soundness, \ie we need to show the coinductive
counterpart to Lemma \ref{lem:step_lam_ind}:

\begin{restatable}{lemma}{lemsteplamcoind}
    \label{lem:step_lam_coind}
    Let $P$ be a logic program and $A$, $B_1$, \dots, $B_n$ atomic formulae.
    If $P, ( ~\Rightarrow B_1), \dots ( ~\Rightarrow B_n)
    \vDash_{coind} A$ then $P \vDash_{coind} B_1, \dots, B_n \Rightarrow A$.
\end{restatable}
\begin{fullproof}{proofsteplamcoind}
    Assume that, for an arbitrary substitution $\sigma$, for all $i$, $\sigma
    B_i$ is valid in $\McP$. Then, for any grounding substitution $\tau$, from
    the definition of the semantic operator and from the assumption of the lemma
    it follows that $(\tau \circ \sigma) A \in \McP$. Therefore, $\sigma A$ is
    valid in $\McP$. The substitution $\sigma$ is chosen arbitrary whence, for
    any $\sigma$, if, for all $i$, $\sigma B_i$ are valid in $P$ then also
    $\sigma A$ is valid in $P$. From the definition of validity it follows that
    $P \vDash_{coind} B_1, \dots, B_n \Rightarrow A$.
    \qed
\end{fullproof}

\noindent
Now, the universal coinductive soundness of extended coinductive type class
resolution follows straightforwardly:

\begin{restatable}{theorem}{thmlpmlamnu}
    \label{thm:lp-m_lam_nu}
    Let $\Phi$ be an axiom environment for a logic program $P$, and let be $\Phi \vdash e : F$ for a
    formula F by the \ref{ir:lp-m}, \ref{ir:lam}, and \ref{ir:nu} rules.
    Then $P \vDash_{coind} F$.
\end{restatable}
\begin{shortproof}
    By 
    induction on the derivation tree using Lemmata
    \ref{lem:sem_prop}, \ref{lem:step_nu}, \&~\ref{lem:step_lam_coind}.
    \qed
\end{shortproof}
\begin{fullproof}{prooflpmlamnu}
    By structural induction on the derivation tree.

    Base case: Let the derivation be with empty assumptions. Then it is by the
    rule \ref{ir:lam} and of the form $\inference[]{}{\Phi \vdash \kappa : \sigma
    A}$ for an atomic formula $A$, a constant symbol $\kappa$, and a
    substitution $\sigma$. By definition of the rule \ref{ir:lp-m} there is a
    clause  $(\kappa : ~\Rightarrow A) \in \Phi$.
    By Lemma~\ref{lem:sem_prop} c), $\Phi \vDash_{coind} \sigma A$.

    Inductive case: Let the last step be by the rule \ref{ir:lp-m} and of
    the form $\inference[]{\Phi \vdash e_1 : \sigma B_1
    \dots \Phi \vdash e_n : \sigma B_n}{\Phi \vdash \kappa e_1 \dots e_n :
    \sigma A}$
    for an atomic formulae $A$, $B_1$, \dots, $B_n$ a constant symbol $\kappa$,
    a substitution $\sigma$ and proof witnesses $e_1$, \dots, $e_n$.  By
    definition of the rule \ref{ir:lp-m} there is a clause $(\kappa : B_1,
    \dots B_n \Rightarrow A) \in \Phi$.
    
    By the induction assumption, for $i \in \{1, \dots, n\}$, $\Phi
    \vDash_{coind} B_i$ and by Lemma~\ref{lem:sem_prop} d), $\Phi
    \vDash_{coind} \sigma A$.

    Let the last step of the derivation be by the rule \ref{ir:lam}. Then it is
    of the form\\ $\inference[]{\Phi, (\beta_1 : ~\Rightarrow B_1), \dots,
    (\beta_n : ~\Rightarrow B_n) \vdash e : A}{\Phi \vdash \lambda \beta_1,
    \dots, \beta_n . e : B_1, \dots, B_n \Rightarrow A}$ for atomic formulae
    $A$, $B_1$, \dots, $B_n$, proof term $e$, and variables $b_1$, \dots, $b_n$.
    By the induction assumption, $\Phi, (\beta_1 : ~\Rightarrow B_1), \dots,
    (\beta_n : ~\Rightarrow B_n) \vDash_{coind} A$
    and by Lemma~\ref{lem:step_lam_coind} also $\Phi \vDash_{coind} B_1, \dots,
    B_n \Rightarrow A$.

    Let the last step be by the rule \ref{ir:nu} and of the form\\
    $\inference[]{\Phi, (\alpha : B_1, \dots, B_n \Rightarrow A) \vdash e : B_1,
    \dots B_n \Rightarrow A}{\Phi \vdash \nu \alpha . e : B_1, \dots, B_n
    \Rightarrow A}$ for an atomic formulae $A$, $B_1$, \dots, $B_n$, a variable
    $\alpha$ and a proof
    witness $e$ in the head normal form. W.l.o.g. let $e = \lambda \beta_1 \dots
    \beta_n . \kappa e_1 \dots e_m$.
    Therefore there is 
    inference step of the form\\
    $\inference[]{
        \Phi, (\beta_1 : ~\Rightarrow B_1), \dots, (\beta_n : ~\Rightarrow B_n)
            , (\alpha : B_1, \dots, B_n \Rightarrow A)
            \vdash e_1 : \sigma C_1'\\
        \dots \\
        \Phi, (\beta_1 : ~\Rightarrow B_1), \dots, (\beta_n : ~\Rightarrow B_n)
            , (\alpha : B_1, \dots, B_n \Rightarrow A)
            \vdash e_m : \sigma C_m'
    }{
        \Phi, (\beta_1 : ~\Rightarrow B_1), \dots, (\beta_n : ~\Rightarrow B_n)
            , (\alpha : B_1, \dots, B_n \Rightarrow A)
            \vdash \kappa e_1 \dots e_n : \sigma A'
    }$\\
    for $\sigma A' = A$ and $(\kappa : C'_1, \dots C'_m \Rightarrow
    A') \in \Phi$. By the induction assumption, for all $i$, $\Phi,
    (\beta_1 : B_1), \dots, (\beta_n : B_n), (\alpha :
    B_1, \dots, B_n \Rightarrow A) \vDash C_i$.
    By Lemma~\ref{lem:step_nu} $\Phi \vDash_{coind} B_1, \dots, B_n \Rightarrow
    A$.
    \qed
\end{fullproof}

\subsubsection{Coinductive Incompleteness of the Proof System \ref{ir:lp-m} + \ref{ir:lam} + \ref{ir:nu}.}
In Section \ref{sec:ind_sound}, we considered two ways of stating inductive completeness of type class resolution. 
We state the corresponding result for the coinductive case here.
We start with a formulation in the traditional style of Lloyd~\cite{Lloyd87}:

\textbf{Coinductive Completeness-1:} 
\emph{if a ground atomic formula $F$
  is in $\MP'$, then  $\Phi_P \vdash e:F$ in the  \ref{ir:lp-m} + \ref{ir:lam} + \ref{ir:nu} proof system.}
Such a result should not hold, since there exist coinductive logic programs that define corecursive schemes that cannot be captured by this proof system.
Consider the following example~\cite{FuKSP16}:
\begin{example}
    Let $\Sigma$ be a signature with a binary predicate symbol $D$, a unary function
    symbol $s$ and a constant function symbol $z$. Consider a program $P_{11}$
    with the signature $\Sigma$ given by the following axiom environment:
    \vskip-1.4em
\begin{alignat*}{3}
    \kappa_1 :~ & \mathtt{D}(x, \mathtt{s}(y)) && \Rightarrow \mathtt{D}(\mathtt{s}(x), y)\\
    \kappa_2 :~ & \mathtt{D}(\mathtt{s}(x), \mathtt{z}) && \Rightarrow \mathtt{D}(\mathtt{z}, x)
\end{alignat*}
\noindent
Let us denote a term $ \mathtt{s}( \mathtt{s}( \dots \mathtt{s}( x)\dots))$
where the symbol $\mathtt{s}$ is applied $i$-times as $\mathtt{s}^i(x)$. By
observing the construction of $\McP$ we can see that, for all $i$,
if $\mathtt{D}(\mathtt{z}, \mathtt{s}^i(x))$ then $\mathtt{D}(\mathtt{s}^i(x),
\mathtt{z})\in \McP$ and also $\mathtt{D}(\mathtt{z}, \mathtt{s}^{i -
1}(\mathtt{x})) \in \McP$. Therefore $\mathtt{D}(\mathtt{z}, \mathtt{z}) \in
\McP$. However, there is no proof of $\mathtt{D}(\mathtt{z}, \mathtt{z})$ since
any number of proof steps resulting from the use of \ref{ir:lp-m} generates yet another ground 
premise that is different from all previous premises. Consequently, the proof cannot be
closed by \ref{ir:nu}.
\newknote{Show the derivation here, if you have space}
\fnote{I do not think we can squeeze it in.}
Also, no lemma that would allow for a proof can be
formulated; an example of such a lemma would be the above $\mathtt{D}(\mathtt{z}, \mathtt{s}^i(x))
\Rightarrow \mathtt{D}(\mathtt{z}, \mathtt{s}^{i - 1}(\mathtt{x}))$.  This is a
higher order formula and cannot be expressed in a first order Horn clause logic.
\end{example}
%

\noindent
It is easy to show that the following kind of completeness theorem
also fails: 

  \textbf{Coinductive Completeness-2:} 
\emph{if  $\MP' \vDash_{coind} F$ then  $\Phi_P \vdash e:F$ in the  \ref{ir:lp-m} + \ref{ir:lam} + \ref{ir:nu} proof system.}

\vspace{1em}
\noindent
Recall Examples \ref{ex:incomp1} and \ref{ex:incomp2}, and the programs $P_6$
and $P_7$. We demonstrated that, in general, there are formulae that are valid
in $\MP$ but do not have a proof in $P$. The same two examples will serve our
purpose here. For example, the greatest Herbrand model of the program $P_6$ is $\McP = \MP = 
\{A(g), A(f(g)), A(f(f(g))), \dots \}$. Therefore, for an atomic formula
$\mathtt{A}(x)$, $P \vDash_{coind} \mathtt{A}(x)$. However, it is impossible to construct 
a proof:
\begin{prooftree}
\vskip-1.5em
        \footnotesize
        \AxiomC{$\cdots$}
        \RightLabel{\sc Lp-m}
        \UnaryInfC{$
            P \vdash
            e : \mathtt{A}(x)
        $}
\end{prooftree}
\vskip-0.5em
The rules \ref{ir:lp-m} and \ref{ir:lam} are not applicable here for the same
reasons as in the inductive case and the rule \ref{ir:nu} is not applicable
since $\mathtt{A}(x)$ is not a Horn formula.
\newknote{not a Horn formula? }
\fnote{as defined in the Definition \ref{df:syntax}}
\newknote{Do you mean to say: because $e$ must be in the head normal form, which would require application of the rule  \ref{ir:lp-m}.}
\newknote{I am removing the below text -- it is very repetitive}


\vskip-2em~
\subsubsection{Related Program Transformation Methods}

We conclude this section with a discussion of program transformation with
Horn formulae that are entailed by the rules \ref{ir:lam} and \ref{ir:nu}.
From the fact that the \ref{ir:nu'} rule is inductively unsound, 
it is clear that using program transformation techniques based on the lemmata that were proved by the \ref{ir:lam} and \ref{ir:nu} rules
would also be inductively unsound.
However,  a more interesting result is that adding such program clauses will not
change the coinductive soundness of the initial program:

\begin{restatable}{theorem}{thmcoindtrans}
    \label{thm:pt}
  Let $\Phi$ be an axiom environment for a logic program $P$, and let  $\Phi \vdash e : F$ for a 
  formula $F$ by the \ref{ir:lp-m}, \ref{ir:lam} and  \ref{ir:nu} rules such
  that $\HNF(e)$. Given a formula $F'$, 
  $P \vDash_{coind} F'$ iff  $(P, F)  \vDash_{coind} F'$.
\end{restatable}

\begin{fullproof}{proofcoindtrans}
    By the Theorem \ref{thm:lp-m_lam_nu}, $P \vDash_{coind} F$. Therefore, $\McP$ is
    a model of $F$ and $\McP = \McP {}_{, F}$. Hence $P \vDash_{coind} F'$ iff  $P,
    F \vDash_{coind} F'$.

\end{fullproof}

\noindent
The above result is possible thanks to the head normal form condition, since it is
then impossible to use a clause $A \Rightarrow A$ that was derived from
an empty context by the
rule~\ref{ir:lam}.  It is also impossible to make such a derivation within the
proof term $e$ itself and to then derive $A$ by the  \ref{ir:nu}  rule from $A
\Rightarrow A$. The resulting proof term will fail to satisfy the head normal
form condition that is required by ~\ref{ir:nu}.  Since this condition guards against any
such cases, we can be sure that this program transformation method is
coinductively sound and hence that it is safe to use with any coinductive dialect of
logic programming, e.g.  with CoLP~\cite{SimonBMG07}.
%
    
%

\section{Related Work}
\label{sec:relwork}

The standard approach to type inference for type classes, corresponding to type
class resolution as studied in this paper, was described by Stuckey and
Sulzman~\cite{StuckeyS05}. Type class resolution was further studied by L\"ammel
and Peyton Jones \cite{LammelJ05}, who described what we here call \emph{coinductive type class
resolution}. The description of the extended calculus of
Section~\ref{sec:ext_coind_sound} was first presented by Fu \etal~\cite{FuKSP16}.
Generally, there is a body of work that focuses on allowing
for infinite data structures in logic programming. Logic programming with
rational trees \cite{Colmerauer84,JaffarS86} was studied from both an operational
semantics and a declarative semantics point of view. Simon \etal \cite{SimonBMG07}
introduced \emph{co-logic programming} (co-LP) that also allows for terms that are
rational infinite trees and hence that have infinite proofs.  Appropriate models of
these paradigms are,  respectively, the greatest Herbrand model
\newknote{why least?}
\fnote{typo I presume?}
and the stratified alternating
fixed-point co-Herbrand model. On the other hand, corecursive
resolution, as studied in this paper, is more expressive than co-LP:
while also allowing infinite proofs, and closing of coinductive hypotheses is less constrained in our approach.
\newknote{It is not more expressive: CoLP can deal with programs $p(x) \Rightarrow p(f(x))$, and this calculus cannot. Corecursive resolution is only more expressive in one aspect -- universal proofs. You may as well explain this properly.}

\newknote{Please include comparison to Rosu and Lucanu here, as reviewer is asking: the relevant work section is rather slim. For good measure, you may add a citation to a similar work}
\fnote{I already find it hard to fit the paper within page limit.}


\section{Conclusions and Future Work
}
\label{sec:con}

In this paper, we have addressed three research questions.  First, we provided a
uniform analysis of type class resolution in both inductive and coinductive
settings and proved its soundness relative to (standard) least and greatest
Herbrand models. 
Secondly, we demonstrated, through several examples, that coinductive
resolution is indeed coinductive---that is, it is not sound relative to least
Herbrand models.  Finally, we addressed the question of whether the methods
listed in this paper can be reapplied to coinductive dialects of logic
programming \emph{via} soundness preserving program transformations.

\newknote{I removed sentences saying that you want to prove completeness. We have said all we could about this in this paper, right?}
\newknote{Generally, I would expect more text in this section... }
As future work, we intend to extend our analysis of Horn-clause resolution to
Horn clauses with existential variables and existentially quantified goals. We
believe that such resolution accounts to type inference for other language
constructs than type classes, namely type families and algebraic data types.

\section*{Acknowledgements}


This work has been supported by the EPSRC grant ``Coalgebraic Logic Programming for Type Inference''
EP/K031864/1-2,  EU Horizon 2020 grant ``RePhrase:
Refactoring Parallel Heterogeneous Resource-Aware Applications - a Software
Engineering Approach'' (ICT-644235), and by COST Action IC1202 (TACLe), supported by
COST (European Cooperation in Science and Technology).

\bibliographystyle{splncs03}
\bibliography{references}


\appendix

\newpage 


\section{Omitted Proofs}
\label{sec:appendix}

This appendix gives full versions of the proofs that have been omitted in the body of the paper.

\lemsemprop*
\proofsemprop

\thmindsound*
\proofindsound

\lemsteplamind*
\proofsteplamind

\thmlpmlam*
\prooflpmlam

\thmindtrans*

\newknote{please include the proof of the above theorem in the appendix, for consistency. Even if it is a few lines showing how to apply Lemma 2}

\proofindtrans

\lemstepnuprm* 
\proofstepnuprm

\thmlpmnuprm*
\prooflpmnuprm

\propadmisible*
\proofadmisible

\lemstepnu*
\proofstepnu

\lemsteplamcoind*
\proofsteplamcoind

\thmlpmlamnu*
\prooflpmlamnu

\thmcoindtrans*
\newknote{Proof?}

\proofcoindtrans




\end{document}